\documentclass[11pt]{article}
\usepackage[a4paper, margin=1in]{geometry}
\usepackage{hyperref}
\usepackage{algorithm}
\usepackage{algcompatible}
\usepackage{amsmath}
\usepackage{amssymb}
\usepackage{mathtools}
\usepackage{amsthm}
\usepackage{dsfont}
\usepackage{enumitem}
\usepackage[capitalize,noabbrev]{cleveref}
\usepackage{tikz}
\usepackage{natbib}

\theoremstyle{plain}
\newtheorem{theorem}{Theorem}[section]
\newtheorem{proposition}[theorem]{Proposition}
\newtheorem{lemma}[theorem]{Lemma}

\theoremstyle{definition}
\newtheorem{definition}[theorem]{Definition}

\theoremstyle{remark}

\DeclareMathOperator*{\argmin}{arg\,min}

\DeclareMathOperator{\polylog}{poly\,log}

\DeclareSymbolFontAlphabet{\mathbb}{AMSb}
\newcommand{\hatc}{\widehat{C}}
\newcommand{\pp}{q}

\title{Faster Global Minimum Cut with Predictions}
\author{
Benjamin Moseley\thanks{Supported in part by Google Research Award, an Infor Research Award, a Carnegie Bosch Junior Faculty Chair, NSF grants CCF-2121744 and CCF-1845146.}\\ Carnegie Mellon University\\ \texttt{moseleyb@andrew.cmu.edu}  
\and 
Helia Niaparast\thanks{Supported in part by NSF CCF-1845146.
} \\ Carnegie Mellon University\\ \texttt{hniapara@andrew.cmu.edu}
\and
Karan Singh\\Carnegie Mellon University\\  \texttt{karansingh@cmu.edu}
}
\date{}

\begin{document}

\maketitle

\sloppy

\begin{abstract}
Global minimum cut is a fundamental combinatorial optimization problem with wide-ranging applications. Often in practice, these problems are solved repeatedly on families of similar or related instances. However, the de facto algorithmic approach is to solve each instance of the problem from scratch discarding information from prior instances.  

In this paper, we consider how predictions informed by prior instances can be used to warm-start practical minimum cut algorithms. The paper considers the widely used Karger's algorithm and its counterpart, the Karger-Stein algorithm. Given good predictions, we show these algorithms become near-linear time and have robust performance to erroneous predictions.  Both of these algorithms are randomized edge-contraction algorithms. Our natural idea is to probabilistically prioritize the contraction of edges that are unlikely to be in the minimum cut. 
\end{abstract}
\section{Introduction}
Machine learning is driving scientific breakthroughs. While this has transformed many areas, there remain domains where machine learning holds immense, yet unrealized potential. One such area is the reimagining of computer science foundations through machine learning, particularly for designing \emph{faster discrete algorithms}.

A rapidly growing body of work, collectively referred to as \textbf{algorithms with predictions}, focuses on leveraging machine-learned predictions to overcome worst-case performance barriers. In recent years, hundreds of papers have explored this model, mainly applying it to improve \emph{quality of solutions} produced by online algorithms. The primary challenge in the online setting is uncertainty. Hence, machine learning is naturally suited to this setting. \citet{lykouris2021competitive} provide a theoretical framework that characterizes the competitive ratio of algorithms based on the quality of machine-learned predictions. Subsequent research has applied this model to achieve higher-quality solutions and break worst-case lower bounds in various problems, including caching \citep{im2022parsimonious, lykouris2021competitive}, binary search \cite{DinitzILMNV24}, scheduling \citep{lindermayr2022permutation, lattanzi2020online, im2023non}, and clustering \citep{lattanzi2021robust}. For a comprehensive overview, see the survey by \citet{mitzenmacher2022algorithms}.

Interestingly, the work of \citet{kraska2018case}, which arguably initiated this line of work, had a different goal in mind and emphasized improvements in {\em running time}, a direction that remains underexplored compared to advancements in solution quality. Their empirical results highlight the significant potential of machine learning to accelerate algorithms and motivate exploration of the algorithmic possibilities of using machine learning for runtime improvements.

Traditionally, computer science solves problems from scratch, with running times analyzed using worst-case scenarios. However, in practice, many problems are repeatedly solved over time. The conventional worst-case model often \emph{discards valuable information} shared between instances. Given that problem instances frequently exhibit similarities, machine learning offers an opportunity to learn a \emph{warm-starting state} that can enhance algorithmic performance. The community has begun to investigate theoretical guarantees for algorithms that use machine-learned warm starts to achieve runtime improvements. \citet{dinitz2021faster} initiated this line of inquiry by showing how predicted dual variables could accelerate the Hungarian algorithm for bipartite matching. Building on this idea, researchers have developed runtime-improving algorithms for discrete combinatorial optimization problems, such as shortest paths \citep{chen2022faster,mccauley2025}, maximum flow \citep{davies2023predictive}, list labeling data structures \citep{mccauley2024online}, and dynamic graph algorithms \citep{mccauley2024incremental}.

Despite these advances, this area remains underdeveloped, with significant open questions regarding how machine-learned predictions can improve algorithmic runtimes. 
 
\paragraph{Global Minimum Cut.}
This paper focuses on the global minimum cut problem. Consider an undirected graph $G = (V,E,w)$ with vertex set $V$ and edge set $E$, where each edge $e$ has a nonnegative weight $w(e)$. We will use $m$ to denote the number of edges and $n$ to denote the number of vertices. For an edge subset $E'\subseteq E$, let $w(E')$ be the sum of weights of edges in $E'$, and for a subset $V'\subseteq V$ of vertices, let $\delta(V')$ be the edges between $V'$ and $V\setminus V'$.  The goal is to find a partition $(S, T)$ of the vertex set that minimizes $w(\delta(S)) = \sum_{e \in E\cap (S\times  T)} w(e)$, thus minimizing the total weight of the edges that cross the partition. Occasionally, we will refer to the minimum cut by the edges it contains, rather than by the partition that induces it.

This problem has been extensively studied in the literature \citep{gomory1961multi, hao1994faster, stoer1997simple, karger1993global, karger1996new, karger2000minimum, gabow1995matroid, kawarabayashi2018deterministic, henzinger2020local, saranurak2021simple, li2020deterministic, li2021deterministic, henzinger2024deterministic, chekuri1997experimental}, and has wide-reaching applications, e.g., network design and network reliability \citep{colbourn1987combinatorics}, information retrieval \citep{botafogo1993cluster}, and compiler design for parallel languages \citep{chatterjee1995array}. Following a sequence of breakthrough results, the fastest known algorithms for this problem run in near-linear $O(m\polylog n)$ time, even when constrained to be deterministic \citep{henzinger2024deterministic}. However, known near-linear time algorithms are primarily of theoretical interest and have not been implemented due to their complexity. In fact, popular graph libraries \citep{siek2001boost, hagberg2008exploring} resort to algorithms that are much slower in theory but easier to implement.

\paragraph{Karger's and Karger-Stein Algorithms.}  Karger's algorithm \citep{karger1993global} and its extension, the Karger-Stein algorithm \citep{karger1996new}, are two renowned randomized algorithms for finding the global minimum cut.  They are frequently used as algorithmic benchmarks \citep{chekuri1997experimental}. The practical relevance of Karger's algorithm draws from its simplicity and its highly parallelizable nature. Given an unweighted graph $G = (V,E)$, at each iteration, Karger's algorithm picks an edge $e \in E$ uniformly at random and contracts its endpoints, keeping parallel edges but removing self-loops. Once there are only two vertices left, the partition of the vertex set formed by these two ``metavertices" is returned as a candidate for the minimum cut. This algorithm can be easily extended to weighted graphs, where an edge $e$ is picked with probability $w(e)/w(E)$, and instead of adding parallel edges, the edge weights are summed upon contraction. 

It can be shown that the cut reported by Karger's algorithm is actually a minimum cut of the graph with probability at least $\Omega({1}/{n^2})$. Therefore, by repeating Karger's algorithm $O(n^2)$ times and retaining the best cut across all runs, the algorithm recovers the true minimum cut with constant probability. Each run of Karger's algorithm can be performed in $O(m)$ time; thus the total runtime is $O(mn^2)$. However, given a sufficient number of parallel processors, each of these runs can be performed in parallel with zero intermittent communication.

The Karger-Stein algorithm considerably boosts the probability of obtaining the minimum cut in any given trial to $\Omega(1/\log n)$ at the cost of $O(n^2 \log n)$ runtime per trial. The key observation is that the probability that the minimum cut survives a random edge contraction decreases severely the fewer vertices there are left. Therefore, repeated contractions when there are too many vertices remaining as in Karger's algorithm are wasteful. Instead, starting with $n$ vertices, Karger-Stein contracts edges at random until there are about $n/\sqrt{2}$ vertices left, and then returns the better of the cuts received by making two independent recursive calls on the resultant graph. The net runtime of $O(n^2 \log^2 n)$ is nearly optimal for dense graphs.

The natural questions we ask are: {\em How can Karger's algorithm and the Karger-Stein algorithm be improved using predictions? How error-resilient are the resultant prediction-augmented algorithms? What is the right measure of error in such predictions?}

\subsection{Results}  
This paper improves the runtime performance of Karger's algorithm and the Karger-Stein algorithm using predictions.  The first key question is: what should be predicted? The idea is to predict which edges are more or less likely to be in the global minimum cut. Of course, these predictions could be erroneous. We introduce new variants of these algorithms that robustly use these predictions.

\paragraph{Predictions.} Let $C^* \subseteq E$ be a minimum cut in $G$. Since a cut ultimately is a subset of edges that cross some partition of the set of vertices, let us begin by considering binary predictions for each edge $e \in E$ indicating whether or not $e \in C^*$. Let $\hatc \subseteq E$ be the predicted minimum cut. Note that $\hatc$ may not necessarily be a cut.  

Clearly, any edge in the symmetric set-difference $\hatc \triangle C^*$ constitutes an error. The majority of error measures considered in the algorithms with predictions literature are symmetric (e.g., in \citet{mitzenmacher2022algorithms,dinitz2021faster}); they penalize equally for over- and under- prediction. However, an important feature of our work is to disentangle the impact of these two types of errors on the algorithm's runtime. Concretely, the prediction error can be divided into false negatives $(\eta)$, and false positives $(\rho)$. We define $\eta$ as the ratio of the weight of the edges in the minimum cut but not in the prediction to the weight of the minimum cut and $\rho$ as the ratio of the weight of the edges in the prediction but not in the minimum cut to the weight of the minimum cut: 

\[\eta := \frac{w(C^* \setminus \hatc)}{w(C^*)}, \quad \rho := \frac{w(\hatc \setminus C^*)}{w(C^*)}.\]
We will see that a false negative is far more costly than a false positive in boosting the runtime of Karger-like algorithms. 

In fact, we prove and state our results in a more general setting that cleanly generalizes the above definitions to real-valued predictions. Here, each edge $e$ has an associated prediction $p_e \in [0,1]$, representing the probabilistic prediction of its inclusion in $C^*$. Now, $\eta$ and $\rho$ are defined as: 
\[\eta := \frac{\sum_{e\in C^*} (1-p_e)w(e)}{w(C^*)}, \quad \rho := \frac{\sum_{e \in E\setminus C^*}p_ew(e)}{w(C^*)}.\]

In the analysis, we will only use that $\eta$ and $\rho$ are valid upper bounds for the quantities defined above, and, for simplicity, we assume that $\rho$ is at least $1$.  We note that if there is more than one global minimum cut, $\eta$ and $\rho$ can be defined with respect to any fixed minimum cut. Due to this, our run-time guarantees hold with respect to the minimum cut that gives the best run-time in terms of $\eta$ and $\rho$.

\paragraph{Boosted Karger's Algorithm.} Our intuitive approach to taking advantage of predictions is to tweak the graph so that Karger's algorithm has a higher chance of contracting the edges that are not predicted to be in the minimum cut, so that the minimum cut has a higher chance of surviving. 

A reasonable guess a priori is that the amount of computational work needed to compute the minimum cut (and hence, the runtime) scales linearly in the quality of predictions, e.g. knowing about half of the edges in the minimum cut (given no false positives) reduces the total work by a factor of half. However, we show that the improvement is far more stark, and such predictions can eliminate an entire factor of $n$ from the runtime. Thus, even for fixed prediction quality, the multiplicative speed-up grows with the size of instance.

We prove the following theorem, demonstrating the potential for significant improvement in the running time of Karger's algorithm, provided that the predictions are not too erroneous.
\begin{theorem}\label{thm:boosted-omega}
    For a suitable setting of parameters, given predictions measured by $\eta$ and $\rho$ as defined above, the Boosted Karger's algorithm (\Cref{alg:boosted}) outputs a minimum cut with probability at least 
    $$\Omega\left(\frac{1}{n^{2\eta}{\rho}^{2(1-\eta)}}\right).$$
\end{theorem}

Let us compare this to the $\Omega(1/n^2)$ probability of recovering the true minimum cut in the standard Karger's algorithm. Regardless of the value of $\eta$, which is always in $[0,1]$, the result in \Cref{thm:boosted-omega} is better than Karger's algorithm as long as $\rho \leq o(n)$. Thus, the result shows remarkable resiliency to error: Even if none of the minimum cut edges is in the prediction, and the predicted edges are almost $n$ times as many as the minimum cut, the probability of recovering the minimum cut is no worse than Karger's algorithm.

To see the utility of this result, consider when the error is small, e.g. if $\rho$ is a constant, then the probability of success of Boosted Karger's algorithm is $\Omega\left({1}/{n^{2\eta}}\right)$, which is significantly better than that of Karger's algorithm, $\Omega\left({1}/{n^2}\right)$.

\paragraph{Boosting the Karger-Stein Algorithm.} For Karger-Stein, it is important first to carve out the possible regime of improvement. For dense graphs, that is, if $m=\Theta(n^2)$, Karger-Stein is already nearly optimal. For sparse graphs, the best one may hope for is a near-linear runtime of $O(m)$. Therefore, depending on the quality of the predictions, one may hope to interpolate these. This is what our results deliver.

As we will see, our earlier result relied on improving the minimum cut's probability of surviving a single random edge contraction. The Karger-Stein analysis is not directly well suited to make use of this effect. Instead, we adapt a variant, here eponymously termed FPZ, introduced in \citet{fox2019minimum}, who obtain Karger-Stein-style guarantees for finding minimum cuts in hypergraphs. Their analysis was greatly simplified recently by \citet{karger2021recursive}, which we borrow. The difference in FPZ vs. Karger-Stein is that the former executes a single edge-contraction in each step, but makes a random number of recursive calls; the propensity of these is closely tied to the survival probability of a minimum cut during an edge contraction. Here, in addition to tweaking the graph so that random edge contractions are more likely to contract edges outside the predicted set, we modify the propensity for these recursive calls. In the end, we prove the following.

\begin{theorem}\label{thm:ks-main}
    For a suitable setting of parameters, given predictions measured by $\eta$ and $\rho$ as defined above, the Boosted FPZ algorithm runs in time 
	\[\begin{cases}
    		O(m^{1-\eta}n^{2\eta}\log n) & \text{ if }\rho=O(\sqrt{m}),\text{ and}\\
    		O(\rho^{2(1-\eta)}n^{2\eta} \log n) & \text{ otherwise.}
    	\end{cases}\]
    It outputs a minimum cut with probability at least $\Omega(\frac{1}{\log n})$.
\end{theorem}

For a large and forgiving regime of false positives, when $\rho= O(\sqrt{m})$, the running time multiplicatively interpolates that of Karger-Stein and a near-linear time algorithm, depending on the value of $\eta$. For small $\eta$, it is almost linear-time. In fact, the improvement over Karger-Stein persists regardless of the value of $\eta$ as long as $\rho\leq o(n)$.

We note that our running-time analysis, along with the underlying data structures supporting the implementation, differs from prior work. Importantly, our recurrence analysis is careful as to how many edges are processed in each iteration, effectively amortizing the total work done across multiple levels of recursion.

\paragraph{Learning Near-Optimal Predictions.} We also give a learning algorithm to learn near-optimal predictions from solutions to past instances. Specifically, given a distribution over graphs, from which the learning algorithm can draw samples, we show how near-optimal predictions minimizing average runtime over the distribution may be computed in polynomial time and sample complexity.

\paragraph{Experiments.} We conduct three sets of experiments ranging from synthetic to real datasets.

\paragraph{Limitations.} One limitation of our approach is that the setting of the parameters needed for the theoretical results depends on the knowledge of $\eta$ and $\rho$ (or at least on upper bounds for them). Given a family of instances, it might be possible to conservatively estimate these parameters, but we do not pursue this here. However, in our experiments, we do not assume access to such information and apply the same problem-agnostic parameters uniformly across instances. Our empirical results strongly suggest that the algorithms are insensitive to these parameters for a wide-ranging degree of errors, and this might not be a limitation in practice.

\section{Boosted Karger's Algorithm}
We discuss the Boosted Karger's algorithm, and we prove an improved lower bound for its probability of success. The algorithm has two parameters, a scalar and a threshold. The algorithm {\it boosts} the edges in $E \setminus \hatc$, meaning it multiplies the weights of the edges that fall outside the predicted set by a large scalar and then performs random edge contractions similarly as in Karger's algorithm provided there is a sufficient number of vertices remaining. At this point, each remaining vertex (or properly, {\em metavertex}) corresponds to a subset of the original vertex set. When fewer vertices remain than the specified threshold, our algorithm reverts to the standard (i.e., not {\em boosted}) Karger's algorithm on the remaining metavertices.

\begin{algorithm}[h]
\caption{Boosted Karger's Algorithm}
\label{alg:boosted}
\begin{algorithmic}[1]
\STATE \textbf{Input:} graph $G = (V, E, w)$, predictions $\{p_e\}_{e \in E}$.
\STATE \textbf{Parameters:} scalar $B$, threshold $t$.
\STATE Build $G_B = (V,E,w_B)$, where $w_B(e) = (1 + (B-1)(1-p_e))w(e).$
\WHILE{there are $> t$ vertices left in $G_B$}
    \STATE Pick an edge $\Bar{e}$ with probability $w_B(\Bar{e})/\sum_{e\in E(G_B)} w_B(e)$ and contract it.
\ENDWHILE
\STATE Define $G' = (V', E', w) :=$ subgraph of $G$ induced by the remaining $t$ metavertices in $G_B$.
\WHILE{there are at least 3 vertices left in $G'$}
    \STATE Pick an edge $\Bar{e}$ with probability $w(\Bar{e})/\sum_{e\in E(G')} w(e)$ and contract it.
\ENDWHILE
\STATE\textbf{return} the set of edges in $G$ between the two remaining metavertices.
\end{algorithmic}
\end{algorithm}
We refer to the steps on lines 3-6 of the algorithm above as the {\it first phase}, and the execution of the standard Karger's algorithm on lines 7-10 as the {\it second phase}. For brevity of notation, we define the sequence
\[\pp_i \coloneqq 1 - \frac{1 + (B-1)\eta}{Bi/2 - (B-1)(\rho + (1-\eta))}.\]

We begin the analysis by establishing the survival probability of a fixed minimum cut during a single randomized edge contraction. 
\begin{lemma}\label{lem:single-contraction}
Fix a weighted graph $G$, and let $C^* \subseteq E(G)$ be a minimum cut in $G$, with respect to which $\eta$ and $\rho$ are defined. Now consider a weighted graph $G'=(V,E,w)$ with $k$ vertices obtained by a sequence of edge contractions starting from $G$ such that no edge from $C^*$ has been contracted in any of these contractions. Let $w_B(e) = (1+(B-1)(1-p_e))w(e)$ for all edges $e$ in $E$. Then the probability that none of the edges in $C^*$ is contracted in a single randomized edge contraction in $G'$, where an edge $e$ is chosen with probability $w_B(e)/w_B(E)$, is at least $\pp_{k}$, as long as $k\geq 2\rho+2$.
\end{lemma}
\begin{proof}
    Let us start by observing that the probability that $C^*$ remains intact after a random edge contraction is $1 - w_B(C^*)/w_B(E)$. Now we have 
    \begin{align*}
        w_B(C^*) &= \sum\limits_{e \in C^*} (1 + (B-1)(1-p_e))w(e) = w(C^*)\left( 1 + (B-1)\eta\right).
    \end{align*}
 Additionally, for the surviving edges $E$, we have     
    \begin{align*}
        w_B(E) &= \sum\limits_{e \in E} (1 + (B-1)(1-p_e))w(e) \\
        &=\sum\limits_{e \in E}(B - (B-1)p_e)w(e)\\
        &\geq  Bw(E) -(B-1)(\rho w(C^*) + (1-\eta) w(C^*)),
    \end{align*}
    where the last derivation is an inequality for the sole reason that not all of the original false positive edges may have survived by this stage, that is, in earlier contractions used to arrive at $G'$. Note that by now, since each (meta) vertex $v$  corresponds to a cut in the original graph $G$, $w(\delta(v)) \geq w(C^*)$, where $\delta(v)$ is the set of edges incident on $v$. Since every edge has two vertices, $2w(E)=\sum_{v\in V} w(\delta(v))$. Therefore, we have
    \begin{align*}
        w_B(E) &\geq B \left(\frac{kw(C^*)}{2}\right) -(B-1) \left(\rho w(C^*)+(1-\eta)w(C^*)\right)  \\
        &= w(C^*)\left(\frac{Bk}{2} - (B-1) \left(\rho+(1-\eta)\right)\right).
    \end{align*}
We can now write 
\begin{align*}
    1 - \frac{w_B(C^*)}{w_B(E)} \geq 1 - \frac{w(C^*)\left( 1 + (B-1)\eta\right)}{w(C^*)\left(Bk/2 - (B-1) \left(\rho+(1-\eta)\right)\right)} = \pp_{k}.
\end{align*}   
\end{proof}

Next, we state and prove the following elementary inequality that will be used to prove the main result of this section. 
\begin{lemma}\label{lem:algebraic}
    For all $t \geq 2\rho + 2$, it holds
    \begin{align*}
        \prod\limits_{i = t+1}^n \pp_i
        &\geq \left(\frac{t-2\rho-2}{n} \right)^{2\left(\eta + \frac{1-\eta}{B}\right)}.
    \end{align*}
\end{lemma}
\begin{proof}
We first apply the inequality $1 - x \geq e^{\frac{-x}{1-x}}$, which holds for all $x<1$, to get
    \begin{align*}
        \prod\limits_{i = t+1}^n \left( 1 - \frac{1 + (B-1)\eta}{Bi/2 - (B-1)(\rho + (1-\eta))}\right)
        &\geq \exp\left({-\sum_{i=t+1}^n \frac{1 + (B-1)\eta}{B(i/2 -1) -(B-1)\rho}}\right) \\
        &= \exp\left({-\sum_{i=t-1}^{n-2} \frac{1 + (B-1)\eta}{Bi/2 -(B-1)\rho}}\right).
    \end{align*}
Note that for a non-decreasing function $f$, we have $\int_{L-1}^U f(x)dx \leq \sum_{i=L}^{U} f(i)$. Therefore, we can write
\begin{align*}
    \exp\left({-\sum_{i=t-1}^{n-2} \frac{1 + (B-1)\eta}{Bi/2 -(B-1)\rho}}\right) \geq \exp\left({-\int_{t-2}^{n-2} \frac{1 + (B-1)\eta}{Bx/2 -(B-1)\rho}}dx\right).
\end{align*}
The condition $t \geq 2\rho + 2$ ensures that $f$ is non-decreasing in the desired interval. From here, we just need to carry out the calculations and simplify the expressions: 
\begin{align*}
    \exp\left({-\int_{t-2}^{n-2} \frac{2 + 2(B-1)\eta}{Bx -2(B-1)\rho}}dx\right)
    &= \exp\bigg(-\left(2 + 2(B-1)\eta\right) {\int_{t-2}^{n-2} \frac{dx}{Bx -2(B-1)\rho}}\bigg) \\
    &=\exp\bigg(-\left(\frac{2 + 2(B-1)\eta}{B}\right) \ln\left(Bx - 2(B-1)\rho\right)|_{t-2}^{n-2} \bigg) \\
    &= \left(\frac{B(t-2)-2(B-1)\rho}{B(n-2)-2(B-1)\rho}\right)^{\left(\frac{2 + 2(B-1)\eta}{B}\right)} \\
    &= \left(\frac{B(t-2-2\rho)+2\rho}{Bn-2B-2(B-1)\rho}\right)^{2\left(\eta + \frac{1 - \eta}{B}\right)} \\
    &\geq \left(\frac{B(t-2-2\rho)}{Bn}\right)^{2\left(\eta + \frac{1 - \eta}{B}\right)}.  
\end{align*}\end{proof}

We are now ready to prove the following theorem:
\begin{theorem}\label{thm:boosted-exact}
    Let $C^*\subseteq  E$ be a minimum cut in the weighted input graph $G=(V,E,w)$, with respect to which $\eta$ and $\rho$ are defined. Then, assuming $t \geq 2\rho + 2$, the probability that none of the edges of $C^*$ are contracted in the first phase of \Cref{alg:boosted} is at least $\left(\frac{t-2\rho-2}{n} \right)^{2\left(\eta + \frac{1-\eta}{B}\right)}.$   
\end{theorem}

\begin{proof}
Let $E_i$ be the event that none of the edges of $C^*$ are contracted in step $i$ of the algorithm. At the start of step $i$, there are $n-i+1$ remaining vertices. Now, by \Cref{lem:single-contraction}, we have 
\[\Pr(E_i | \{E_j\}_{j < i}) \geq \pp_{n-i+1}.\]

Let $A$ be the event that none of the edges of $C^*$ are contracted in the first phase of the algorithm, that is, until there are at least $t$ vertices left. 
Then, we have 
\begin{align*}
    \Pr(A)&= \Pr(E_1 \cap E_2 \cap \cdots \cap E_{n-t}) \\
        &= \Pr(E_1)\Pr(E_2|E_1)\cdots \Pr(E_{n-t}|E_1E_2\cdots E_{n-t-1}) \\
        &\geq \prod\limits_{i = t+1}^n \pp_i.
\end{align*}

To conclude the claim, we invoke \Cref{lem:algebraic}.
\end{proof}

Theorem \ref{thm:boosted-omega} can now be obtained from Theorem \ref{thm:boosted-exact} by choosing $t$ to be the smallest integer exceeding $3\rho+2$ and any $B=\Omega(\log n)$, because in the second phase Karger's algorithm ensures that the minimum cut has at least $\Omega(1/t^2)$ probability of continued survival.
\section{Boosting the Karger-Stein Algorithm}
In this section, we present a variant of the FPZ algorithm due to \citet{fox2019minimum} that utilizes predictions to improve over the running time of Karger-Stein.

The standard FPZ algorithm is the following. In the algorithm, $\pp'_n \coloneqq 1-2/n$ is a lower bound on the probability that a fixed minimum cut $C^*$ remains intact after a single random  edge contraction on $n$ vertices, assuming none of its edges have been contracted thus far.

\algloop{With}
\algcloop{With}{Otherwise}
\begin{algorithm}[h]
\caption{$\textsc{FPZ}(G, n)$}
\label{alg:fpz-standard}
\begin{algorithmic}[1]
\STATE \textbf{Input:} graph $G = (V, E, w)$ with $n$ vertices.
\STATE \textbf{Parameters:} branching factor $\pp'_n$.
\IF{$n = 2$}
    \STATE \textbf{return} the set of edges in $G$ between the two remaining metavertices.
\ENDIF
\STATE Pick an edge $\Bar{e}$ with probability $\propto w(\Bar{e})$, that is, with probability $w(\Bar{e})/\sum_{e\in E} w(e)$.
\STATE Contract $\Bar{e}$ in $G$ to produce $G'$.
\STATE $C_1 \gets \textsc{FPZ}(G', n-1)$.
\With{ probability $\pp'_n$,}
    \STATE \textbf{return} $C_1.$
\Otherwise
    \STATE $C_2 \gets \textsc{FPZ}(G, n)$.
    \STATE \quad \textbf{return} the cut from $\{C_1, C_2\}$ with the smaller weight.
\end{algorithmic}
\end{algorithm}

The boosted variant is as follows. 

\algloop{With}
\algcloop{With}{Otherwise}
\begin{algorithm}[h]
\caption{$\textsc{BoostedFPZ}(G, n, p)$}
\label{alg:fpz2}
\begin{algorithmic}[1]
\STATE \textbf{Input:} graph $G = (V, E, w)$ with $n$ vertices, predictions $\{p_e\}_{e \in E}$.
\STATE \textbf{Parameters:} scalar $B$, threshold $t$, branching factor $\pp_n$.
\IF{$n = 2$}
    \STATE \textbf{return} the set of edges in $G$ between the two remaining metavertices.
\ELSIF{$n\leq t$}
    \STATE \textbf{return} FPZ(G, n). {\color{gray} // In other words, run the standard ({\em non-boosted}) FPZ algorithm.}  
\ENDIF
\STATE Let $w_B(e) := (1 + (B-1)(1-p_e))w(e)$ for all $e \in E$.
\STATE Pick an edge $\Bar{e}$ with probability $w_B(\Bar{e})/\sum_{e\in E} w_B(e)$ and contract it to produce $G'$.
\STATE $C_1 \gets \textsc{BoostedFPZ}(G', n-1, p)$.
\With{ probability $\pp_n$,}
    \STATE \textbf{return} $C_1.$
\Otherwise
    \STATE $C_2 \gets \textsc{BoostedFPZ}(G, n, p)$.
    \STATE \quad \textbf{return} the cut from $\{C_1, C_2\}$ with the smaller weight.
\end{algorithmic}
\end{algorithm}
In the algorithm above, $\pp_n$ represents a lower bound on the probability that a fixed minimum cut $C^*$ remains intact after a single random boosted edge contraction on $n$ vertices, assuming none of its edges have been contracted thus far. Given prediction $p_e \in [0,1]$ for each edge $e$, we apply our previous idea of reweighting the edges with parameter $B$ to encourage the contraction of edges that lie outside the predicted set. We also add a switching point $t$ to the algorithm, as we did before. Whenever fewer than $t$ vertices remain in the graph, where $t \geq 3\rho + 2$, the algorithm then invokes the standard FPZ algorithm. This time, we use
\[\pp_n \coloneqq 1 - \frac{ 1 + (B-1)\eta}{Bn/2 - (B-1) \left(\rho+(1-\eta)\right)},\]
where $\eta$ and $\rho$ are defined with respect to $C^*$, for $n > t$, and $\pp_n \coloneqq 1-2/n$ for $n \leq t$. We refer to this modified version of the FPZ algorithm as the Boosted FPZ. First, we establish the following lower bound on the success probability of the Boosted FPZ algorithm. 

\begin{theorem}\label{thm:BFPZ-success}
    For any threshold $t$ satisfying $t\geq 3\rho + 2$, the probability that \Cref{alg:fpz2} returns a minimum cut is $\Omega\left(\frac{1}{\log t + \eta \log (n/t)}\right)$.
\end{theorem}
\begin{proof}
Let $C^*$ be the minimum cut in $G$ that is used to define $\eta$ and $\rho$, and let $P(i)$ denote a lower bound on the probability that the algorithm returns $C^*$, given that all edges of $C^*$ have survived contractions up to the point where $i$ vertices are left. Once there are fewer than $t$ remaining vertices, \Cref{alg:fpz2} proceeds identically to the standard FPZ algorithm. Thus, utilizing the result from \citet{karger2021recursive}, we have that $P(t) = 1/(2H_t-2)$.

For the first phase, that is, when there are at least $t$ remaining vertices, we will be able to reuse the following recurrence for $P(n)$ from \citet{karger2021recursive}, although the value of the branching factor (compare $\pp'_n$ and $\pp_n$) is now different.
    \[P(n)= \pp_n^2P(n-1) + (1-\pp_n) (1-(1-P(n)) (1-\pp_n \cdot P(n-1))).\]
As observed in \Cref{lem:single-contraction}, $\pp_n$ is a lower bound on the probability that $C^*$ survives yet another randomized edge contraction. The recurrence is derived by noting that with probability $\pp_n$, line 12 is executed, after which the probability of returning a minimum cut is at least $\pp_n \cdot P(n-1)$. With the remaining probability $1-\pp_n$, both paths to computing the minimum cut through recursive calls on instances of size $n$ and $n-1$ must fail for the algorithm to miss the minimum cut.

This recurrence can be simplified to:
    \[\frac{1}{P(n)} = \frac{1}{P(n-1)} + 1 - \pp_n.\]
Unrolling the recurrence, we get
    \begin{align*}
        \frac{1}{P(n)} &= \frac{1}{P(t)} + \sum\limits_{i = t+1}^n (1 - \pp_i) \\
        &= 2H_t-2 + \sum\limits_{i = t+1}^n \frac{1 + (B-1)\eta}{Bi/2 - (B-1) \left(\rho+(1-\eta)\right)} \\
        &\leq 2H_t-2 + \int_{t}^{n} \frac{1 + (B-1)\eta}{Bi/2 - (B-1) \left(\rho+(1-\eta)\right)}di \\
        &= 2H_t-2 + \frac{2\left(1 + (B-1)\eta\right)}{B}\ln \left(\frac{Bn/2 - (B-1)(\rho + (1-\eta))}{Bt/2 - (B-1)(\rho + (1-\eta))}\right) \\
        &= O\left(\log t + \eta \log \frac{n}{t}\right),
    \end{align*}
where we use the fact that $\int_{L-1}^U f(x)dx \leq \sum_{i=L}^{U} f(i)$ holds for any non-decreasing function $f$, and in particular for $f(x)=-1/x$. This concludes the proof.
\end{proof}

Next, we will analyze the running time of the algorithm. Let us take a moment to revisit a textbook implementation of Karger's algorithm that runs in near-linear time using Kruskal's algorithm. Recall, Kruskal's algorithm is typically used for finding minimum spanning trees. 

In the unweighted case, we start by creating a uniformly random permutation of all edges and processing them sequentially. Throughout the algorithm, we use a union-find data structure to check which nodes have been merged. This is used in the standard fast implementation of the traditional Karger's algorithm.  When processing an edge in this order, like in Kruskal's algorithm, any edge with both endpoints in the same connected component is discarded.  If an edge's endpoints are in different components, the union-find data structure merges these nodes. The partition of vertices  formed just before merging the last two connected components, is returned as the minimum cut.
 This approach can also be extended to the weighted case, e.g., using the Gumbel trick.

A similar implementation can  be performed for the Boosted FPZ algorithm. This will in turn enable the efficient run-time of our algorithm. The implementation of each random edge contraction must be carried out in two cases, each utilizing a different data structure. The data structure used is based on the number of remaining vertices, $n$. 

If $n > t$, a union-find data structure is used and the edges are sampled lazily without replacement, with probabilities proportional to their boosted weights. This is done until an edge is found whose endpoints belong to different components. For sampling, a categorical distribution over edges can be maintained online, for example, using a red-black tree, while allowing sampling without replacement in $O(\log m)$ time.  To make recursive calls on the same graph, it is too inefficient to copy the graph and run recursive calls separately. Instead, the algorithm runs one call and then later returns to a possibly second recursive call, in a depth-first manner over the recursion tree. We utilize a union-find data structure with deletions, which also takes $O(\log m)$ time in the worst case per operation \citep{AlstrupTGRZ14}. 

If $n \leq t$, we switch to an adjacency list data structure, which allows a random edge contraction in time proportional to the number of remaining vertices, as suggested in \cite{karger2021recursive}. When switching between these two regimes, we prune the list of remaining edges in $O(m \log m)$ time to ensure that there are at most $t^2$ remaining edges. Once we are in the second phase, \Cref{alg:fpz2} is identical to the FPZ algorithm, the total run-time thereafter is $O(t^2 \log t)$.  

\begin{proof}[Proof of \Cref{thm:ks-main}]
    Let $T(k, \ell)$ be an upper bound on the expected running time of the algorithm on any call with $k$ vertices and $\ell$ edges left to process. Since we switch to the standard FPZ algorithm at $t$ vertices, for all $\ell'$, we have $T(t,\ell') = O(t^2\log t)$, which is the expected running time of the FPZ  \citep{karger2021recursive}. Note that, as mentioned above, we can assume $\ell' \leq t^2$ in this case.   
    
    Then, we have the following recurrence for $T(k, \ell)$, carefully considering the number of edges processed in each iteration, via the union-find data structure, to carry out one edge contraction. The recursive expression can be explained as follows: in any call, the algorithm processes $\ell - \ell'$ edges for some $\ell'$, taking $(\ell - \ell')\cdot O(\log n)$ time. The algorithm then makes a recursive call on $k-1$ vertices and $\ell'$ edges. Furthermore, with probability $(1-\pp_k)$, the algorithm repeats itself on the input graph. For the analysis, we take the maximum over all possible values of $\ell'$ to consider the worst case for the algorithm.
    \begin{align*}
       T(k,\ell) &\leq \max\limits_{1 \leq \ell' < \ell} \{T(k-1, \ell') + (1-\pp_k)T(k,\ell) + (\ell - \ell')\cdot O(\log n)\}. 
    \end{align*}
    This inequality can be simplified to:
    \[T(k,\ell) \leq \frac{1}{\pp_k}\max\limits_{1 \leq \ell' < \ell} \{T(k-1, \ell') + (\ell - \ell')\cdot O(\log n)\}.\]
    Unfolding the right-hand side, we get:
    \begin{align*}
        T(n,m)&\leq \max\limits_{\ell_i}\left\{ \sum\limits_{i = t+1}^{n} \frac{\ell_i \cdot O(\log n)}{\prod_{i \leq j \leq n} \pp_j} : \sum_{i = t+1}^n \ell_i \leq m\right\}+ O\left(\frac{t^2\log t}{\prod_{t+1 \leq j \leq n} \pp_j}\right) \\
        &\leq \frac{O(m\log n+ t^2\log t)}{\prod_{t+1 \leq j \leq n} \pp_j}.
    \end{align*}
    Using Lemma \ref{lem:algebraic} to lower bound the product of $\{\pp_n\}$, by setting $B=\Omega(\log n)$, for any $t\geq 3\rho+2$, we get $$
        T(n, m) = O\left(\frac{m \log n + t^2\log t}{(t/n)^{2\eta}}\right).$$
    Setting $t=\max\{\lceil3\rho+2\rceil, \sqrt{m}\}$ concludes the claim.
\end{proof}
\section{Learning Near-Optimal Predictions}
In this section, we describe how near-optimal predictions can be learned from past data. Formally, we assume that there is an unknown fixed distribution $\mathcal{D}$ on weighted graphs that share the same set $V$ of vertices, from which a number of samples are drawn independently. Given these samples, our goal is to learn near-optimal predictions $p^* \in [0,1]^{\binom{V}{2}}$ that minimize the expected runtime of the Boosted Karger's algorithm with respect to $\mathcal{D}$.

Let $C^*(G)$ denote a minimum cut in $G$. For a prediction $p$, let $\eta(G,p)$ and $\rho(G,p)$ denote the false negative and false positive with respect to $C^*(G)$ and $p$, respectively. Note that we can assume, without loss of generality, that $0 \leq {w}_G(e) \leq 1$ for all $e \in E(G)$. Otherwise, we can scale all edge weights so that they are within the interval $[0,1]$, and this would not change $C^*, \eta,$ and $\rho$. Since the edge sets of the sampled graphs may differ, we will assume the vector of predictions $p$ is defined over $\binom{V}{2}$.   
 
We have established that the expected running time of the Boosted Karger's algorithm is at most 
$R(G,p) := n^{2\eta(G,p)}\rho(G,p)^{2(1 - \eta(G,p))}$. A natural strategy for learning near-optimal predictions is to compute predictions that minimize this running time upper bound averaged over collected samples. However, $R(G,p)$ is nonconvex in $p$. Instead, we aim to optimize $U(G,p):= n^{2\eta(G,p)}\tilde{\rho}(G,p)^2$. Here $\tilde{\rho}$ is a variation of $\rho$ defined as:
\[\tilde{\rho}(G,p) := \frac{({\bf 1} - w_*(G))^\top  p}{w_G(C^*(G))},\]
where $w_{*}(G)$ is the characteristic weight vector of the minimum cut $C^*(G)$, that is, it is a vector in $[0,1]^{\binom{V}{2}}$, with its entry corresponding to $e$ equal to $w_G(e)$ if $e \in C^*(G)$, and zero otherwise. Note that $\tilde{\rho}(G,p) \geq \rho(G,p)$ for all $p$, and therefore $U(G,p)$ is a valid upper bound on $R(G,p)$. It is instructive to compare $\rho$ and $\tilde{\rho}$ for unweighted graphs. For unweighted graphs, $\rho(G,p)=\sum_{e\in E(G)\setminus C^*(G)} p_e/w_G(C^*(G))$ and $\tilde{\rho}(G,p)= \sum_{e\in \binom{V}{2}\setminus C^*(G)} p_e/w_G(C^*(G))$. Thus, the principal difference between the two is that $\tilde{\rho}(G,p)$ additionally accounts for erroneous predictions that correspond to missing edges. 

Unfortunately, $U(G,p)$ is also not convex in $p$ (see \Cref{prop:U-not-convex}). Our key observation is that upon replacing $\mathbf{1}^\top p$, which appears naturally in the definition of $\tilde{\rho}$, with a free variable, the resultant analogue of $U(G,p)$ becomes convex in $p$ (see \Cref{prop:Ub-convex}). Since $\mathbf{1}^\top p$ is an instance-independent scalar, its best value can be estimated through a grid search, in addition to running copies of online gradient descent on $p$ corresponding to all possible values of the free variable in the grid. In Appendix~\ref{sec:app}, we prove:
\begin{theorem}\label{thm:learning}
For any $\varepsilon, \delta > 0$, there exists an algorithm with the following properties. Given $\textrm{poly}(n,1/C_{\min},\log (1/\varepsilon\delta))/\varepsilon^2$ i.i.d samples from any distribution $\mathcal{D}$, satisfying that $C_{min}$ is a lower bound on the size of the minimum cut of any graph in $\mathcal{D}$'s support, the algorithm runs in $\textrm{poly}(n, 1/\varepsilon,1/C_{\min},\log (1/\delta))$ time and outputs a prediction $\bar{p}\in [0,1]^{\binom{V}{2}}$ such that, with probability at least $1-\delta$, we have
\[\mathbb{E}_{G \sim \mathcal{D}}\left[U(G,\bar{p})\right] -  \argmin_{p\in [0,1]^{\binom{V}{2}}} \mathbb{E}_{G \sim \mathcal{D}}\left[U(G,p)\right] \leq \varepsilon.\]
\end{theorem}

The time and sample complexity above scale with $1/C_{\min}$. Such dependencies occur regularly while learning real-valued functions without uniformly bounded derivatives (see, e.g., Chapter 4 in \cite{hazan2016introduction}). For unweighted graphs, one may always assume that $C_{min}\geq 1$.
\section{Experiments}
We aim to demonstrate that the theoretical advantages presented in the previous sections also translate to improved empirical performance.
We perform three sets of experiments \footnote{The Python implementation of the experiments is available at \href{https://github.com/helia-niaparast/global-minimum-cut-with-predictions}{https://github.com/helia-niaparast/global-minimum-cut-with-predictions}.}. 
The first involves synthetic settings where we can explicitly control the fidelity of predictions to study the performance of the algorithm quantitatively. The second is a setting where Karger's algorithm is used to repeatedly find minimum cuts on a family of instances that organically arise from trying to solve traveling salesperson (TSP) instances. We conclude with experimental comparisons on real data.

\subsection{Controlled Experiments on Synthetic Graphs}
In the following set of experiments, we explicitly control the prediction quality and measure the performance of the proposed algorithm against Karger's algorithm on a family of synthetically generated graphs.  We are especially interested in:
\begin{enumerate}[itemsep=0em,topsep=1em]
    \item How the Boosted Karger's algorithm compares to the original Karger's algorithm, especially on instances where the latter requires many repetitions to succeed.
    \item How the asymmetric error measures $\eta$ and $\rho$ affect the number of trials that the Boosted Karger's algorithm needs to find the minimum cut. 
\end{enumerate}

A bipartite graph $G$ with $n$ vertices is built as follows. Each partite set has $n/2$ vertices, and the edges consist of random perfect matchings. First, $k$ random perfect matchings are added to the graph, and then an arbitrary vertex is picked and $\ell$ of its incident edges are randomly chosen and removed from the graph. These instances are designed to be difficult for Karger-like algorithms because they contain many near-minimum-cut-sized cuts, each of which has a healthy probability of survival via random edge contractions.

To generate predictions, we first compute the true minimum cut $C^*$ on $G$. Now for any given $\eta$ and $\rho$, we pick two random subsets $C_{\eta} \subseteq C^*$ and $C_{\rho} \subseteq E \setminus C^*$, such that $w(C_{\eta}) = \eta w(C^*)$ and $w(C_{\rho}) = \rho w(C^*)$. Our predicted edge set is $\hatc_{\eta, \rho} = C^* \setminus C_{\eta} \cup C_{\rho}$. 

We build $G$ with $n = 600, k = 100, \ell = 10$. We note the number of trials that Karger's algorithm needs on $G$ to find the minimum cut. This is our baseline. Next, we fix a value of $\rho \in \{0, 10, 100\}$, and for each value of $\eta \in \{0, 0.05, 0.1, \dots, 1\}$, we measure the number of trials Boosted Karger's needs to find the minimum cut with input $(G, \hatc_{\eta, \rho})$ with $(B,t)=(n,2)$. In \Cref{fig:matching-eta}, this process is repeated $30$ times.

We can see that the Boosted Karger's algorithm outperforms Karger's algorithm by two orders of magnitude when $\eta \leq 0.5$ and $\rho \in \{0, 10\}$. Furthermore, even for $\rho = 100$, indicating especially poor prediction quality, since the predicted set of edges is about a hundred times as numerous as the size of the minimum cut, Boosted Karger's algorithm is better by one order of magnitude when $\eta \in [0, 0.6]$. 

\begin{figure*}[h]
    \begin{center}
        \includegraphics[width=0.32\textwidth]{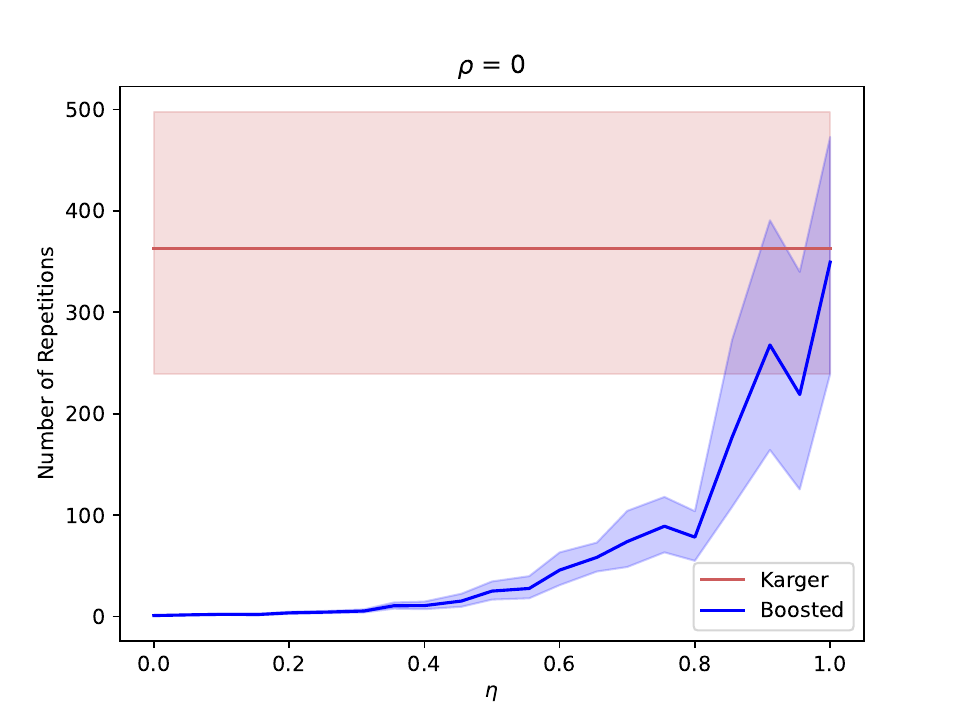}
        \includegraphics[width=0.32\textwidth]{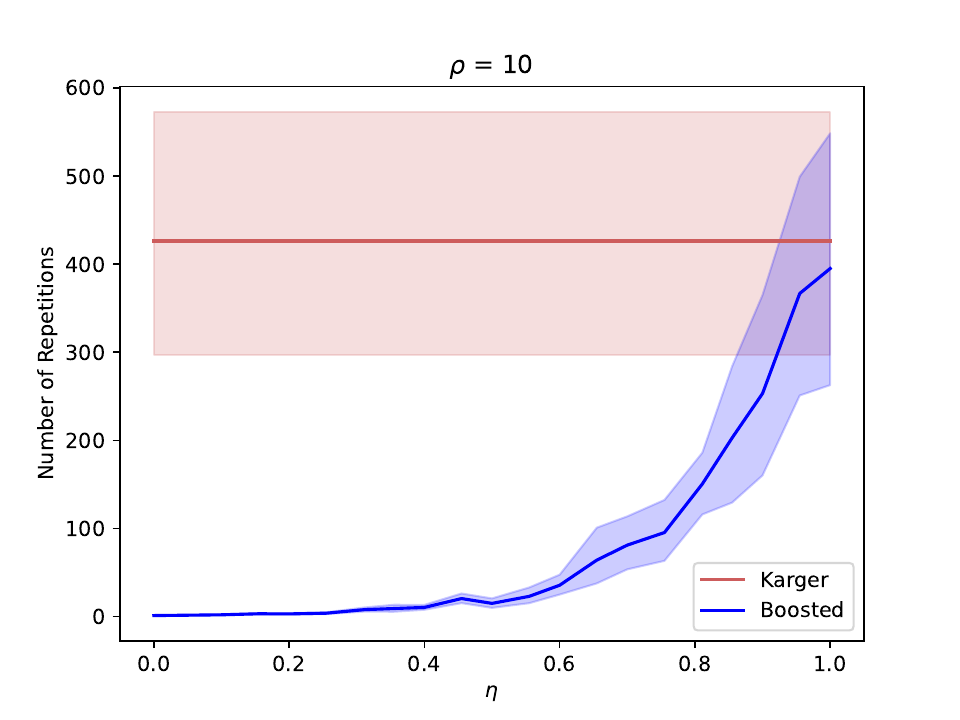}
        \includegraphics[width=0.32\textwidth]{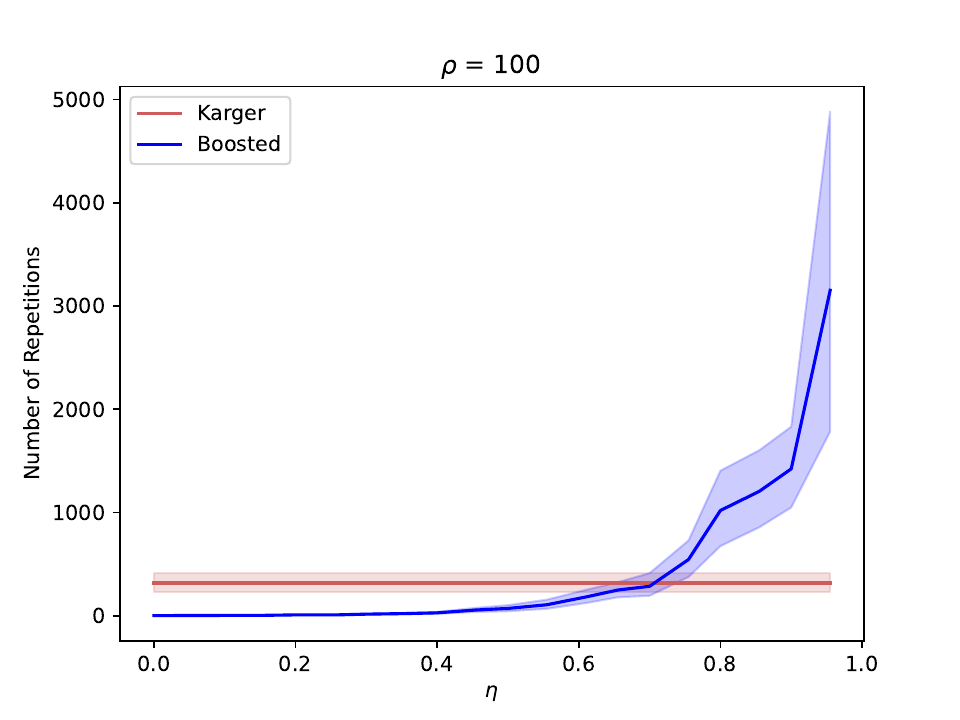}
    \caption{A controlled experimental comparison of the number of repetitions Boosted Karger's algorithm needs to find the mincut vs. the standard Karger's algorithm for different quality of predictions, as parameterized by $\eta$ and $\rho$.}
    \label{fig:matching-eta}
    \end{center}
\end{figure*}

\subsection{Minimum Cut Instances from the TSP LP}
In the second set of experiments, we explore instances in which predictions appear naturally. Cutting plane algorithms for the traveling salesperson problem (TSP) proceed by repeatedly identifying subtour elimination constraints in the subtour linear program relaxation for TSP, the search for which can be recast as finding global minimum cuts (see, e.g., \cite{chekuri1997experimental}). We use this practical use case of Karger's algorithm to evaluate the performance of the Boosted Karger's algorithm. 

The subtour elimination approach to TSP starts by minimizing $\sum_{e \in E} w_e x_e$ subject to $\sum_{e \in N(v)} x_e = 2$ for all nodes $v \in V$ and $0 \leq x_e \leq 1$ for all $e \in E$, to obtain an initial solution $x^0$. This linear program is known as the subtour relaxation.     
Then, a new graph $G_0$ is built with the same set of vertices and edges as the original, but with the difference that the weight of edge $e$ in $G_0$ is $x^0_e$. Note that if the entire vector $x^0$ is integral, then $x^0$ represents a Hamiltonian cycle, and the size of the minimum cut in $G_0$ is 2. 

The issue is that the edges maybe fractional and the goal is to find a constraint to add to the program based on $x^0$. Upon finding a minimum cut in $G_0$, if its size is smaller than 2, the following subtour elimination constraints are added to the above LP and the LP is solved again. 
\[\sum\limits_{\substack{e = \{u,v\} \\ u,v \in S_1}} x_e \leq |S_1| - 1, \text{ and} 
\sum\limits_{\substack{e = \{u,v\} \\ u,v\in S_2}} x_e \leq |S_2| - 1,\]
where $(S_1, S_2)$ is the vertex partition for the minimum cut found in $G_0$. This process is repeated until the minimum cut in the current graph is of weight 2. Thus, global minimum cuts are used to generate constraints for the subtour relaxation. 

To begin, we first construct a TSP instance. A random graph $G = (V, E)$ with $n$ vertices is built as follows. The vertices are partitioned into two subsets $S$ and $T$ of equal size, and the edge set consists of a number of random cycles. First, $k$ random Hamiltonian cycles are added to $G$ with the guarantee that each cycle crosses the partition $(S,T)$ in exactly two edges. Then, $k$  random cycles of length $n/2$ are added to each of $S$ and $T$. Finally, $\varepsilon k$ smaller random cycles, each having a random length between 3 and $n/2-1$, are added to $G$, making sure that these cycles do not cross the partition.

We obtain a sequence of graphs $G_0, G_1, \ldots, G_\ell$, for which we want to find the minimum cut. We predict that none of the edges with integer weights appear in the minimum cut. Therefore, for each graph $G_i$, the predicted minimum cut $\hatc_i$ is the set of all edges with fractional weights. These are natural and easily computable predictions. 

We build $G$ with $n = 500, k = 50, \varepsilon = 0.5$, and construct $G_0, \ldots, G_\ell$. Then, we do the following steps for each $i \in [\ell]$. On each minimum cut instance we obtain, we measure the number of iterations needed to produce the minimum cut for Karger's and for Boosted Karger on $(G_i, \hatc_i)$. We set $(B,t) =(\log n,2)$. These steps are repeated 10 times. 

In \Cref{fig:TSP}A, we evaluate both algorithms and observe that the Boosted Karger's algorithm consistently outperforms Karger's algorithm. In particular, it achieves an order-of-magnitude improvement on the harder instances where Karger's algorithm requires many repetitions to find the minimum cut.

\begin{figure*}[h]
    \centering
    \includegraphics[width=0.43\linewidth]{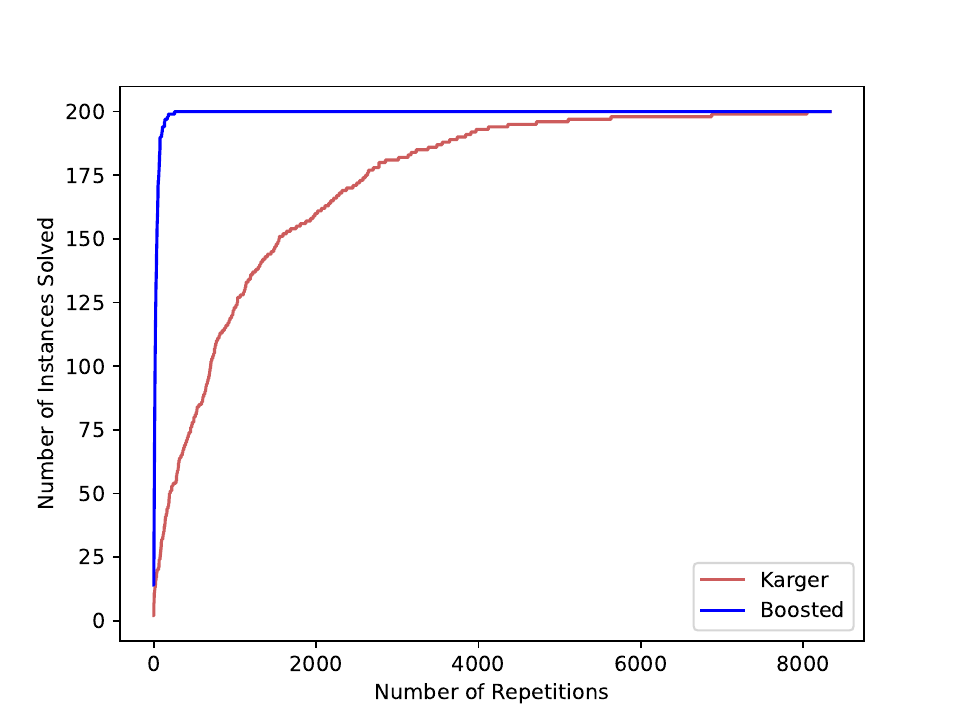}
    \includegraphics[width=0.49\linewidth]{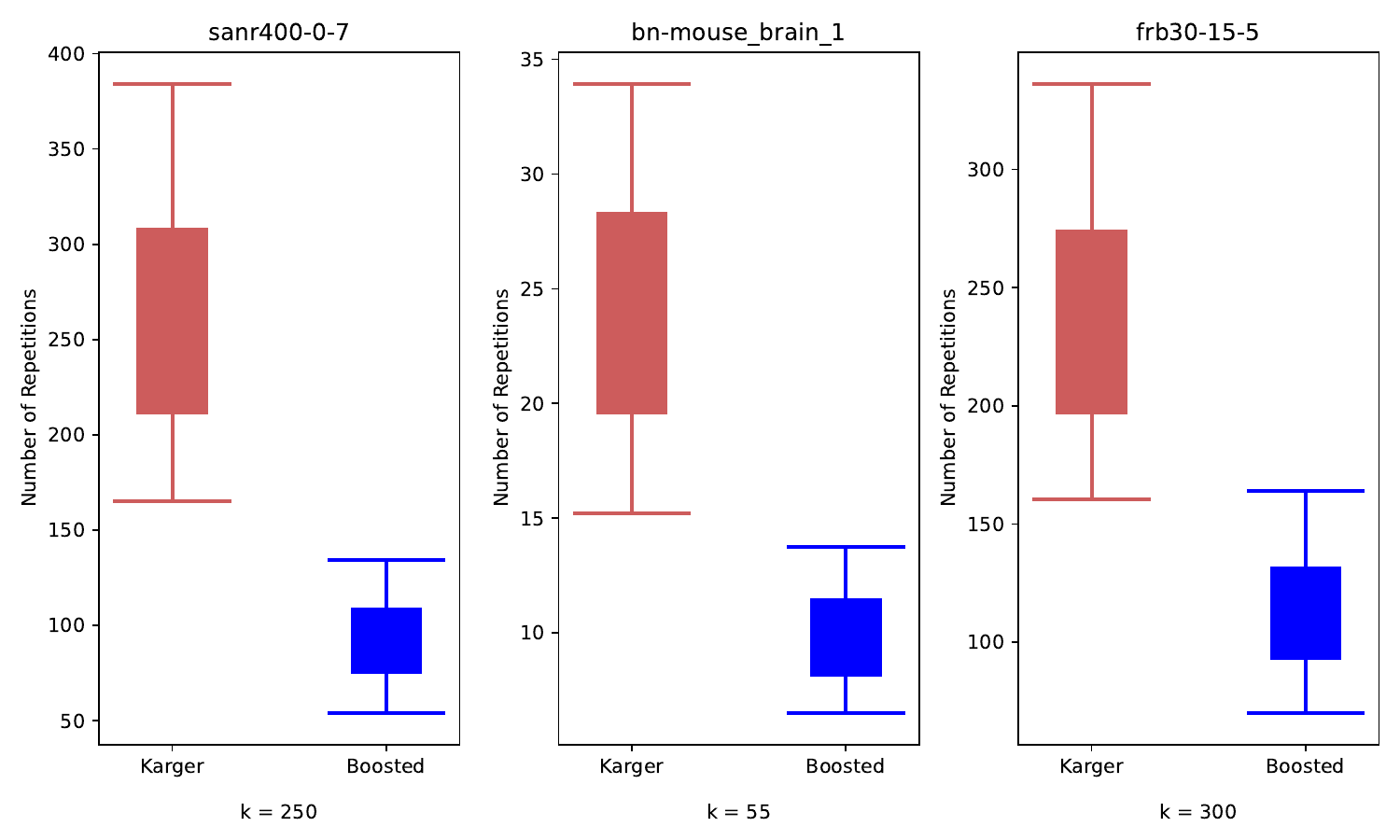}
    \caption{In \Cref{fig:TSP}A, on the left, we compare the number of minimum cuts arising from the subtour TSP that Karger's and Boosted Karger's algorithms  solved within a given number of repetitions. On the right, in \Cref{fig:TSP}B, is demonstrated the number of repetitions needed to recover the minimum cut on three real graph datasets.}
    \label{fig:TSP}
\end{figure*}

\begin{figure}[h!]
    \centering

    \label{fig:real}
\end{figure}

\subsection{Real Datasets}
Finally, we compare the performance of the Boosted Karger's algorithm and the standard variant on three real datasets from \citet{nr-aaai15}. For each dataset, the predictions are obtained by first randomly sampling half of the edges and then performing $k$ parallel runs of Karger's algorithm on the sampled edges. The predicted edge set is formed by the union of the edges of the $k$ cuts found by Karger's algorithm. As a heuristic, we pick $k$ to be close to the minimum degree of the graph. This process is repeated $100$ times in \Cref{fig:TSP}B. We observe that for all three datasets the Boosted Karger's algorithm requires discernibly fewer number of trials to find the minimum cut.
\section{Conclusion}
We explored how predictions about the minimum cut can be impactful in boosting the performance of Karger's and the Karger-Stein algorithms. Furthermore, we empirically demonstrated that the Boosted Karger's algorithm outperforms Karger's algorithm even when predictions have a relatively high error. The paper shows Karger's algorithm can naturally be improved with predictions, and a natural direction for future research is to explore how predictions may be applied to speed up other combinatorial optimization problems.

\bibliography{refs}
\bibliographystyle{plainnat}

\newpage
\appendix
\section{Learning Near-Optimal Predictions}\label{sec:app}

We begin by proving the following two propositions, which motivate our learning algorithm.

\begin{proposition}\label{prop:U-not-convex}
    The function $U(G,p)$ is not convex in $p$. 
\end{proposition}

\begin{proof}
The hessian of $U(G,p)$ w.r.t $p$ is as follows:
\begin{align*}
  \nabla^2U(G,p) &= \frac{2n^{2\eta(G,p)}}{w_G(C^*(G))^2} 
  \Bigg( \big({\bf 1} - w_*(G) - 2\ln n \cdot \tilde{\rho}(G,p) \cdot w_*(G) \big) \\
  &\quad \times \big({\bf 1} - w_*(G) - 2\ln n \cdot \tilde{\rho}(G,p) \cdot w_*(G) \big)^T \\
  &\quad - 2 \big(\ln n \cdot \tilde{\rho}(G,p) \cdot w_*(G)\big) 
  \big(\ln n \cdot \tilde{\rho}(G,p) \cdot w_*(G)\big)^T 
  \Bigg).
\end{align*}
Being a difference of two rank-one matrices, the hessian is not positive semi-definite. Concretely, consider the following simple example demonstrating that $\nabla^2U(G,p)$ is not positive semi-definite, which means that $U(G,p)$ is not convex in $p$. Consider the following graph on 3 vertices with the edge weights written next to them. Let $p(\{1,2\}) = 1/\ln 3$, and $p({\{2,3\}}) = p({\{1,3\}}) = 0$. 
\begin{center}
    \begin{tikzpicture}[thick, scale = 1.5]
    \filldraw (0,0) circle (0.05) node[inner sep = 5pt, left]{1};
    \filldraw (1,1) circle (0.05) node[inner sep = 5pt, above]{3};
    \filldraw (1,0) circle (0.05) node[inner sep = 5pt, below]{2};

    \draw (0,0) -- (1,0) node [midway, below]{0.6};
    \draw (1,0) -- (1,1) node [midway, right]{0.7};

\end{tikzpicture}
\end{center}
Then, we have
\[\nabla^2U(G,p) = \frac{2n^{2\eta(G,p)}}{w_G(C^*(G))^2}\begin{pmatrix}
    -0.16 & -0.4 & -0.4 \\
    -0.4 & 1 & 1 \\
    -0.4 & 1 & 1 
\end{pmatrix},\]
which is not positive semi-definite.
\end{proof}

Nevertheless, we will succeed in minimizing $U(G,p)$ over $p$. To build towards this, consider the following function formed by replacing $\mathbf{1}^\top p$ in $U(G,p)$ by a free variable $b$.

\begin{definition}
   For $b \geq 0$, define 
    \[U^b(G,p) := n^{2\eta(G,p)}\left(\frac{b - \langle w_*(G), p \rangle}{w_G(C^*(G))}\right)^2, \text{and}\quad K_b := \{p \in [0,1]^{\binom{V}{2}}:{\bf 1}^Tp \leq b\}.\] 
\end{definition}

We have the following proposition.

\begin{proposition}\label{prop:Ub-convex}
    For all $b \geq 0$, $U^b(G,p)$ is convex in $p$ over $K_b$.
\end{proposition}
\begin{proof}
    The hessian of $U^b(G,p)$ w.r.t $p$ is as follows:
\begin{align*}
    \nabla^2U^b(G,p) &= \frac{n^{2\eta(G,p)}}{w_G(C^*(G))^2} \cdot 
    \left( 4\ln^2 n \left(\frac{b - \langle w_*(G), p\rangle}{w_G(C^*(G))}\right)^2  
    + 8\ln n \left(\frac{b - \langle w_*(G), p\rangle}{w_G(C^*(G))}\right) + 2 \right) \\
    &\quad \times w_*(G) w_*(G)^T.
\end{align*}

The coefficient of $w_*(G)w_*(G)^T$ in the expression above is non-negative whenever $p \in K_b$. Therefore, the hessian is positive semi-definite on the interior of $K_b$, which means $U^b(G,p)$ is convex over $K_b$.
\end{proof}

Now, we describe our proposed algorithm to compute a prediction $\bar{p}$ given a polynomial number of i.i.d samples drawn from $\mathcal{D}$, and then analyze it to prove \Cref{thm:learning}.

\paragraph{The Learning Algorithm.}
To begin, we draw $T$ i.i.d samples $G_1, \ldots, G_T$ from $\mathcal{D}$. 
We discretize the range of possible values for $b$ (which represents $\mathbf{1}^\top p$), i.e., $[0, {\binom{n}{2}}]$, into equally sized intervals, and optimize $p$ over each of them separately. 
Let $\mathcal{B}$ be the set of discrete values considered for $b$; we will specify the resolution of the grid $\mathcal{B}$ later.

For each $b \in \mathcal{B}$, we perform online gradient descent on the sequence $\{U^b(G_t,\cdot)\}_{t = 1}^T$ of convex functions over the convex body $K_b$ to obtain the set of vectors $\{p_t^b\}_{t=1}^T \subseteq K_b$, as stated below:
$$ p^b_{t+1} = \Pi_{K_b}\left[p^b_t - \eta_t \nabla U^b(G_t, p^b_t)\right], $$
where $\Pi_{K_b}[x]= \argmin_{y\in K_b} \|x-y\|_2$. Let $\bar{p}_b := \frac{1}{T}\sum_{t = 1}^T p_t^b$ for all $b \in \mathcal{B}$.

Next, we draw $T'$ new i.i.d samples $G'_1, \ldots, G'_{T'}$ from $\mathcal{D}$ and compute $\frac{1}{T'}\sum_{t = 1}^{T'} U^b(G'_t, \bar{p}_b)$ for each $b \in \mathcal{B}$. Let $b' = \argmin_{b \in \mathcal{B}}  \frac{1}{T'}\sum_{t = 1}^{T'} U^b(G'_t, \bar{p}_b)$. The algorithm outputs $\bar{p}_{b'}$. 
The values of $\eta_t$, $T$, $T'$, and the cardinality of $\mathcal{B}$ will be determined in the analysis. 

We start the analysis with the following guarantee:
\begin{theorem}[Theorem 3.1 in \citet{hazan2016introduction}]\label{thm:OGD}
  For a fixed $b \in \mathcal{B}$, let $Q$ be an upper bound on $\|\nabla U^b(G_t, p)\|_2$ for all $(t,p) \in [T] \times K_b$, and let $D$ be an upper bound on $\|p - q\|_2$ for all $p,q \in K_b$. The iterates produced by Online Gradient Descent with step sizes $\eta_t = D/Q\sqrt{t}$ guarantee that:
  \[\sum\limits_{t = 1}^T U^b(G_t, p^b_t) - \min_{p \in K_b} \sum_{t = 1}^T U^b(G_t,p) \leq \frac{3}{2}QD\sqrt{T}.\]
\end{theorem}

To use the above theorem, note that 
\[\nabla U^b(G,p) = -\frac{2n^{2\eta(G,p)}}{w_G(C^*(G))}\left[\ln n \left(\frac{b - \langle w_*(G), p\rangle}{w_G(C^*(G))}\right)^2 + \left(\frac{b - \langle w_*(G), p\rangle}{w_G(C^*(G))}\right)\right]w_*(G).\]

Let $C_{min} := \inf_{G \in \textrm{supp}(\mathcal{D})} w_{G}(C^*(G))$. Then, the guarantee in \Cref{thm:OGD} is obtained by setting $Q = {2n^7\ln n}/{C_{min}^3}$, $D = n$; these valid upper bounds on size of the gradients and the diameter can be readily verified. 

For each $b \in \mathcal{B}$, let $\text{Regret}^b_T :=\sum_{t = 1}^T U^b(G_t, p^b_t) - \min_{p \in K_b} \sum_{t = 1}^T U^b(G_t,p)$. Now, we utilize the following theorem:
\begin{theorem}[Theorem 9.5 in \citet{hazan2016introduction}] \label{thm:OCO-to-stat}
    For a fixed $b \in \mathcal{{B}}$ and any $\delta > 0$, given $T$
    i.i.d samples drawn from $\mathcal{D}$, with probability at least $1 - \delta$, we have 
    \[\mathbb{E}_{G \sim \mathcal{D}}[U^b(G,\bar{p}_b) - U^b(G,p^*_b)] \leq \frac{\text{Regret}^b_T}{T} + \sqrt{\frac{8 \log \left(\frac{2}{\delta}\right)}{T}},\]
    where $p^*_b = \argmin_{p \in K_b} \mathbb{E}_{G\sim \mathcal{D}}[U^b(G,p)]$.
\end{theorem}
Let $p^* = \argmin_{p \in K}\mathbb{E}_{G \sim \mathcal{D}}[U(G,p)]$, where $K = [0,1]^{\binom{V}{2}}$, and let $b^* = {\bf 1}^Tp^*$. Let $\tilde{b} \in \mathcal{B}$ be the smallest element in $\mathcal{B}$ that is larger than or equal to $b^*$, and suppose $\tilde{b} - b^* \leq \Delta$. 
Let $L(G,p)$ be the Lipschitz constant of $n^{2\eta(G,p)}\left(\frac{b - \langle w_*(G), p \rangle}{w_{G}(C^*(G))}\right)^2$ as a function of $b$.
Then, we have
\[\mathbb{E}_{G \sim \mathcal{D}}[U^{\tilde{b}}(G,p^*)] \leq \mathbb{E}_{G \sim \mathcal{D}}[U^{b^*}(G,p^*)] + L\Delta = \mathbb{E}_{G \sim \mathcal{D}}[U(G,p^*)] + L\Delta,\]
where $L$ is an upper bound on $L(G,p)$ for all $G$ in $\mathcal{D}$'s support and $p \in K$. 
Note that $\Delta$ and $L$ can be chosen such that $\Delta \leq n^2/|\mathcal{B}|$ and $L \leq 2n^4/C_{min}^2$.

Let $M$ be an upper bound on $|U^b(G,p) - U^b(G,q)|$ for all $b \in \mathcal{B}$, $p,q \in K_b$, and $G$ in $\mathcal{D}$'s support.
We can pick $M$ such that $M \leq n^6/C_{min}^2$. 
Then, by the Chernoff-Hoeffding inequality and union bound, if $T' = \Theta\left(\left(\frac{M}{\varepsilon}\right)^2\log \left(\frac{|\mathcal{B}|}{\delta}\right)\right)$, then with probability at least ${1-\delta}$, for all $b \in \mathcal{B}$, we have
\[\left|\mathbb{E}_{G \sim \mathcal{D}}\left[U^b(G,\bar{p}_b)\right] - \frac{1}{T'}\sum_{t = 1}^{T'} U^b(G'_t,\bar{p}_b)\right| \leq \varepsilon.\]

Therefore, using the fact that $b'$ is chosen to minimize the empirical average, we have
\[\mathbb{E}_{G\sim \mathcal{D}}\left[U^{b'}(G,\bar{p}_{b'})\right] \leq \frac{1}{T'}\sum_{t = 1}^{T'} U^{b'}(G'_t,\bar{p}_{b'}) + \varepsilon \leq \frac{1}{T'}\sum_{t = 1}^{T'} U^{\tilde{b}}(G'_t,\bar{p}_{\tilde{b}}) + \varepsilon \leq \mathbb{E}_{G \sim \mathcal{D}}\left[U^{\tilde{b}}(G,\bar{p}_{\tilde{b}})\right] + 2\varepsilon.\]

Finally, by Theorem \ref{thm:OCO-to-stat}, we have
\[\mathbb{E}_{G \sim \mathcal{D}}\left[U^{\tilde{b}}(G,\bar{p}_{\tilde{b}})\right] \leq \mathbb{E}_{G \sim \mathcal{D}}\left[U^{\tilde{b}}(G,{p}_{\tilde{b}}^*)\right]+ \frac{\text{Regret}^{\tilde{b}}_T}{T} + \sqrt{\frac{8 \log \left(\frac{2}{\delta}\right)}{T}}.\]

Putting everything together, we can write
\begin{align*}
    \mathbb{E}_{G \sim \mathcal{D}}\left[U^{b'}(G,\bar{p}_{b'})\right] &\leq \mathbb{E}_{G \sim \mathcal{D}}\left[U^{\tilde{b}}(G,\bar{p}_{\tilde{b}})\right] + 2\varepsilon \\
    &\leq \mathbb{E}_{G\sim \mathcal{D}}\left[U^{\tilde{b}}(G,{p}_{\tilde{b}}^*)\right]+ \frac{\text{Regret}^{\tilde{b}}_T}{T} + \sqrt{\frac{8 \log \left(\frac{2}{\delta}\right)}{T}} + 2\varepsilon\\
    &\leq\mathbb{E}_{G \sim \mathcal{D}}\left[U^{\tilde{b}}(G,p^*)\right] + \frac{\text{Regret}^{\tilde{b}}_T}{T} + \sqrt{\frac{8 \log \left(\frac{2}{\delta}\right)}{T}} +2\varepsilon\\
    &\leq \mathbb{E}_{G \sim \mathcal{D}}\left[U(G,p^*)\right] + L\Delta +\frac{\text{Regret}^{\tilde{b}}_T}{T} + \sqrt{\frac{8 \log \left(\frac{2}{\delta}\right)}{T}} + 2\varepsilon.
\end{align*}

Therefore, we need $|\mathcal{B}| = \Theta(n^6/\varepsilon C_{min}^2)$ and 
$T = \Theta\left(\max\left\{\left(\frac{QD}{\varepsilon}\right)^2, \frac{1}{\varepsilon^2}\log \frac{1}{\delta}\right\}\right)$ 
to obtain the promised guarantee.

\end{document}